\documentclass[sigconf]{aamas}
\renewcommand\footnotetextcopyrightpermission[1]{} 
\makeatletter
\renewcommand\@formatdoi[1]{\ignorespaces}
\makeatother

\usepackage{soul}
\usepackage[utf8]{inputenc}
\usepackage{subcaption}
\usepackage{graphicx}
\usepackage{amsmath}
\usepackage{amsthm}
\usepackage{booktabs}
\usepackage{tabularx}
\usepackage{amsfonts}
\usepackage{mathtools}
\usepackage{tikz}
\usepackage{tcolorbox}
\usepackage{kbordermatrix}
\usetikzlibrary{positioning}
\usepackage[ruled,vlined,onelanguage,noend]{algorithm2e}
\usepackage{bm}
\graphicspath{ {./plots/} }

\usepackage{enumitem}
\usepackage{subfiles}
\def\leftcite{\@up[}\def\rightcite{\@up]}

\def\cite{\def\citeauthoryear##1##2{\def\@thisauthor{##1}%
             \ifx \@lastauthor \@thisauthor \relax \else##1, \fi ##2}\@icite}
\def\shortcite{\def\citeauthoryear##1##2{##2}\@icite}

\def\citeauthor{\def\citeauthoryear##1##2{##1}\@nbcite}
\def\citeyear{\def\citeauthoryear##1##2{##2}\@nbcite}

\def\@icite{\leavevmode\def\@citeseppen{-1000}%
 \def\@cite##1##2{\leftcite\nobreak\hskip 0in{##1\if@tempswa , ##2\fi}\rightcite}%
 \@ifnextchar [{\@tempswatrue\@citex}{\@tempswafalse\@citex[]}}
\def\@nbcite{\leavevmode\def\@citeseppen{1000}%
 \def\@cite##1##2{{##1\if@tempswa , ##2\fi}}%
 \@ifnextchar [{\@tempswatrue\@citex}{\@tempswafalse\@citex[]}}

\def\@citex[#1]#2{%
  \def\@lastauthor{}\def\@citea{}%
  \@cite{\@for\@citeb:=#2\do
    {\@citea\def\@citea{;\penalty\@citeseppen\ }%
     \if@filesw\immediate\write\@auxout{\string\citation{\@citeb}}\fi
     \@ifundefined{b@\@citeb}{\def\@thisauthor{}{\bf ?}\@warning
       {Citation `\@citeb' on page \thepage \space undefined}}%
     {\csname b@\@citeb\endcsname}\let\@lastauthor\@thisauthor}}{#1}}

\def\@biblabel#1{\def\citeauthoryear##1##2{##1, ##2}\@up{[}#1\@up{]}\hfill}

\def\@up#1{\leavevmode\raise.2ex\hbox{#1}}

\usepackage{balance} 

\setcopyright{ifaamas}
\acmConference[AAMAS '24]{Proc.\@ of the 23rd International Conference
on Autonomous Agents and Multiagent Systems (AAMAS 2024)}{May 6 -- 10, 2024}
{Auckland, New Zealand}{N.~Alechina, V.~Dignum, M.~Dastani, J.S.~Sichman (eds.)}
\copyrightyear{2024}
\acmYear{2024}
\acmDOI{}
\acmPrice{}
\acmISBN{}

\acmSubmissionID{42}

\title{Efficient Method for Finding Optimal Strategies in Chopstick Auctions with Uniform Objects Values}

\author{Stanisław Kaźmierowski}
\affiliation{
  \institution{University of Warsaw}
  \city{Warsaw}
  \country{Poland}}
\email{s.kazmierowski@uw.edu.pl}

\author{Marcin Dziubiński}
\affiliation{
  \institution{University of Warsaw}
  \city{Warsaw}
  \country{Poland}}
\email{m.dziubinski@mimuw.edu.pl}

\begin{abstract}
We propose an algorithm for computing Nash equilibria (NE) in a class of conflicts with multiple battlefields with uniform battlefield values and a non-linear aggregation function. By expanding the symmetrization idea of Hart~\cite{hart}, proposed for the Colonel Blotto game, to the wider class of symmetric conflicts with multiple battlefields, we reduce the number of strategies of the players by an exponential factor. We propose a \emph{clash matrix algorithm} which allows for computing the payoffs in the \emph{symmetrized} model in polynomial time. Combining symmetrization and clash matrix algorithm with the double oracle algorithm we obtain an algorithm for computing NE in the models in question that achieves a significant speed-up as compared to the standard, LP-based, approach. We also introduce a heuristic to further speed up the process. Overall, our approach offers an efficient and novel method for computing NE in a specific class of conflicts, with potential practical applications in various fields.
\end{abstract}


\newtheorem{observation}{Observation}

\DeclareMathOperator{\maxass}{max\_assign}
\DeclareMathSymbol{\shortminus}{\mathbin}{AMSa}{"39}

\newcommand{\aggfunc}{f\!:\!\left\{\shortminus 1,0,1 \right\}^n\!\rightarrow\!R}

\newcommand{\sign}{\mathrm{sign}}
\newcommand{\playerA}{\mathsf{A}}
\newcommand{\playerB}{\mathsf{B}}
\newcommand{\thisPlayer}{C}
\newcommand{\thatPlayer}{{\shortminus C}}
\newcommand{\examplePlayer}{\thisPlayer \!\in \!\left\{ \playerA, \playerB \right\}}
\newcommand{\thisRes}{D_\thisPlayer}

\newcommand{\aRes}{D_\playerA}
\newcommand{\bRes}{D_\playerB}
\newcommand{\fieldsSet}{\left[n\right]}
\newcommand{\pure}{\bm{s}}
\newcommand{\thisPure}{\pure^\thisPlayer}
\newcommand{\thatPure}{\pure^\thatPlayer}
\newcommand{\aPure}{\pure^\playerA}
\newcommand{\bPure}{\pure^\playerB}
\newcommand{\purePair}{ \aPure, \bPure }
\newcommand{\pureProfile}{\left( \purePair \right)}
\newcommand{\thisPureSet}{S_\thisPlayer}
\newcommand{\thatPureSet}{S_\thatPlayer}
\newcommand{\aPureSet}{S_\playerA}
\newcommand{\bPureSet}{S_\playerB}

\newcommand{\pureSym}{\sigma\!\left(\pure\right)}

\newcommand{\aPureSym}{\sigma\!\left(\aPure\right)}
\newcommand{\bPureSym}{\sigma\!\left(\bPure\right)}
\newcommand{\thisSymSet}{\sigma\!\left(S_\thisPlayer\right)}

\newcommand{\aSymSet}{\sigma\!\left(S_\playerA\right)}
\newcommand{\bSymSet}{\sigma\!\left(S_\playerB\right)}
\newcommand{\symmPair}{\aPureSym,\bPureSym}
\newcommand{\symmProfile}{\left(\symmPair \right)}
\newcommand{\thisPurePayoff}{\pi^\thisPlayer_f \!}
\newcommand{\thisPayoff}{\Pi^\thisPlayer_f \!}
\newcommand{\thisSymPayoff}{\thisPayoff \symmProfile}
\newcommand{\aSymPayoff}{\Pi^\playerA_f \! \symmProfile}

\newcommand{\mixed}{\bm{\xi}}
\newcommand{\thisMixed}{\mixed^\thisPlayer}
\newcommand{\thatMixed}{\mixed^\thatPlayer}
\newcommand{\aMixed}{\mixed^\playerA}
\newcommand{\bMixed}{\mixed^\playerB}
\newcommand{\mixedPair}{\aMixed, \bMixed}
\newcommand{\mixedProfile}{\left(\mixedPair\right)}
\newcommand{\thisMixedSet}{\Delta\left(\thisPureSet\right)}
\newcommand{\thatMixedSet}{\Delta\left(\thatPureSet\right)}

\newcommand{\mixedSym}{\sigma\!\left(\mixed\right)}

\newcommand{\thatMixedSym}{\sigma\!\left(\thatMixed\right)}
\newcommand{\aMixedSym}{\sigma\!\left(\aMixed\right)}
\newcommand{\bMixedSym}{\sigma\!\left(\bMixed\right)}
\newcommand{\mixedSymProfile}{\left(\aMixedSym, \bMixedSym\right)}
\newcommand{\aSinglePayoff}{u_i^\playerA \! \left( \aPure, \bPure \right)}
\newcommand{\bSinglePayoff}{u_i^\playerB \! \left( \aPure, \bPure \right)}
\newcommand{\aPurePayoff}{\pi^\playerA_f \!}
\newcommand{\bPurePayoff}{\pi^\playerB_f \!}
\newcommand{\aPurePayoffArgs}{\aPurePayoff \pureProfile}
\newcommand{\bPurePayoffArgs}{\bPurePayoff \pureProfile}

\newcommand{\thisMixedPayoff}{\thisPayoff \mixedProfile}

\newcommand{\clashM}{\bm{M}}
\newcommand{\clashMatrix}{\clashM^{\aPure, \bPure}}

\newcommand{\outcomesVector}{\bm{u}}
         
\newcommand{\BibTeX}{\rm B\kern-.05em{\sc i\kern-.025em b}\kern-.08em\TeX}

\makeatletter
\gdef\@copyrightpermission{
	\begin{minipage}{0.3\columnwidth}
		\href{https://creativecommons.org/licenses/by/4.0/}{\includegraphics[width=0.90\textwidth]{figs/by.eps}}
	\end{minipage}\hfill
	\begin{minipage}{0.7\columnwidth}
		\href{https://creativecommons.org/licenses/by/4.0/}{This work is licensed under a Creative Commons Attribution International 4.0 License.}
	\end{minipage}
	\vspace{5pt}
}
\makeatother

\begin{document}


\pagestyle{fancy}
\fancyhead{}


\maketitle 

\section{Introduction}
Real-life scenarios of rivalry between two or more players often involve competition in more than one area at a time. Examples include firms competing simultaneously for multiple markets, military commanders deploying troops to the front lines, political parties distributing campaign funds among different regions of the country, or airport security having to assign a number of police dogs among security checkpoints. Game theoretic models used to formally capture such competition form a class called 
\emph{conflicts with multiple battlefields}~\cite{multiple_battlefields}. A famous example of such a model is the Colonel Blotto game~\cite{borel}.
In these models, two players distribute their limited resources across a number of battlefields. Assignment of resources determines the outcome at each battlefield and the outcome of the game is an aggregate of the outcomes at the individual battlefields. The way the outcomes at individual battlefields are aggregated depends on a particular application. A well-known aggregation function is taking the sum of outcomes across all the battlefields. This leads to the Colonel Blotto game.

In this paper we address the problem of computing Nash equilibria under two natural aggregation functions called \emph{more than opponent} and \emph{majoritarian}. The more than opponent aggregation takes value $1$, if the number of battlefields won is strictly greater than the number of battlefields lost, $0$, if it is equal, and $-1$, if it is less. It is suitable for military conflicts where the fighting parties care not about how many battles they win but about winning more battles than the opponent in order to win the war. It is also applicable to political competitions where candidates are rather win-motivated than vote-motivated, caring predominantly about winning the competition and not about the margin at which they win.
The majoritarian aggregation takes a value of $1$, if the number of battlefields won exceeds half of the total number, $0$ if it is equal, and $-1$ if it is less.
It is applicable to single-member districts voting rivalry between two political parties. It is also considered in a model of simultaneous standard auctions with externalities, called chopstick auctions~\cite{three_objects_symm_auction,beyond_chopstick,Englmaier}, where only winning more than half of the objects bears any value to each bidder.
Unlike in the Colonel Blotto game, the aggregation functions in question do not result in bilinear payoffs and, therefore, we can not apply the polynomial time algorithms designed for conflicts with multiple battlefields and bilinear payoffs~\cite{duels_to_battlefields,BehnezhadDehghaniDerakhshanHajiaghayiSeddighin17}.

In this paper, we give a homogeneous formulation of conflicts with multiple battlefields. The main obstacles in computing the NE in these models are the exponential (wrt the model parameters) size of the strategy sets of the players and non-bilinear payoffs. Generalizing the idea of Hart~\cite{hart}, we use a reduction of the strategy sets of the models in question by an exponential factor, called \emph{symmetrization}.
The reduction of the strategy sets makes the computation of payoffs in the symmetrized model challenging. Our main contribution is a \emph{clash matrix algorithm} 
which allows for computing the payoffs in the \emph{symmetrized} models in polynomial time. We combine the polynomial time algorithm for computing payoffs
with the \emph{double oracle algorithm}~\cite{double_oracle_algorithm}, used to compute equilibria in zero-sum games with large strategy spaces.
We propose a heuristic, based on the natural monotonicity properties of the aggregation functions, that allows us to further speed up this algorithm. We show, via computational experiments, that 
thus obtained algorithm achieves a significant speed-up as compared to any approach that requires calculating the entire payoff matrix.

\subsection{Related literature} 
Conflicts with multiple battlefields, considered since the beginning of modern game theory, were first introduced by Borel~\cite{borel}, where the Colonel Blotto game is defined. Since then, 
many game theoretic models that fall into this category were studied, including the hide-and-seek game~\cite{VonNeumann53}, chopstick auctions~\cite{three_objects_symm_auction}, audit games~\cite{BehnezhadBlumDerakhshanHajiaghayiMahdianPapadimitriouRivestSeddighinStark18}, as well as some types of security games~\cite{PitaJainOrdonezPortwayTambeWesternParuchuriKraus09,KorzhykYinKiekintveldConitzerTambe11}. See~\cite{multiple_battlefields} for an excellent survey of this type of models.

The literature on NE computation is vast, and we restrict attention to the closest related works. The models we consider are two-player zero-sum games. 
The computation of NE in such games can be reduced to solving an LP problem~\cite{zero_sum_games_solving} and requires polynomial time with respect to the number of strategies of both players. An exponential number of pure strategies with respect to the model parameters makes those methods inapplicable in the case of conflicts with multiple battlefields. In the case of bilinear payoffs, like in Colonel Blotto game, NE can be computed in polynomial time~\cite{duels_to_battlefields,BehnezhadDehghaniDerakhshanHajiaghayiSeddighin17}. This is also possible for some types of security games~\cite{KorzhykYinKiekintveldConitzerTambe11}.

The \emph{double oracle algorithm (DOA)} was first proposed as a method for solving zero-sum games in~\cite{double_oracle_algorithm}. It was later used for solving zero-sum security games on graphs in~\cite{doa_graph}, which showed that it is worth considering as an alternative for LP-solver-based algorithms.
Games considered in this paper fall into a broader category of integer programming games (IPGs)~\cite{koppe2011}. \cite{CarvalhoLodiPedroso2022} proposed a \emph{sampled generation method (SGM)} which is a significant generalization of DOA to multiple players and arbitrary payoffs. They show that SGM allows for finding NE in a finite number of steps. SGM (and, in particular, DOA) proceeds by iteratively expanding the sets of strategies of the players and solving thus obtained sampled games. Since conflicts with multiple battlefields we consider are zero-sum games, we solve the sample games using the LP solvers facilitated with the clash matrix algorithm. In the case of the general IPGs the PNS algorithm can be used~\cite{PorterNudelmanShoham2008,CarvalhoLodiPedroso2022}.

The rest of the paper is organized as follows. In Section~\ref{sec:model} we define the model of conflicts with multiple battlefields. Section~\ref{sec:complexity} provides the complexity result of finding the best response in chopstick auctions and majoritarian auctions. In Section~\ref{sec:symmetrized} we define the symmetrized version of the model. Computation of payoffs from symmetrized strategy profiles and the clash matrix algorithm are given in Section~\ref{sec:clash_matrix}. Section~\ref{sec:doa} presents the double oracle algorithm and the heuristic for improving its performance in the case of models in question. We provide a computational evaluation of our approach in Section~\ref{sec:evaluation} and conclude in Section~\ref{sec:concl}.

\section{Model}
\label{sec:model}
Conflicts with multiple battlefields model a competition between two players, whom we will denote $\playerA$ and $\playerB$. Each player has a number of discrete resources, e.g. military units or coins, described by $\aRes$ and $\bRes$, respectively. Players compete on the set $\fieldsSet =\left \{1,2,\ldots\,n\right\}$ of $n$ battlefields. We assume that $n \geq 2$. Each battlefield $i \in \fieldsSet$ has value~$1$. The player's strategy is to distribute her resources across battlefields. The set of strategies of player $\thisPlayer \in \left\{ \playerA, \playerB \right\}$ is:
\begin{equation}
\label{pure_strategy}
    \thisPureSet = \left\{ \pure \in \mathbb{N}^n : \sum^n_{i=1} s_i = \thisRes \right\}.
\end{equation}
A pair of strategies $(\purePair)$ is called a \emph{strategy profile}.

Payoff from a single battlefield $i \in [n]$ is defined as
\begin{equation}
\label{pure_single_payoff}
    \aSinglePayoff = -\bSinglePayoff = \sign \! \left( s^\playerA_i - s^\playerB_i \right),
\end{equation}
meaning that a battlefield is won by the player who assigned more resources to it.

Applying~\eqref{pure_single_payoff} to a strategy profile we get a battlefield outcomes vector
\begin{align}\label{outcomes_vector}
    u^\playerA \! \pureProfile = -u^\playerB \! \pureProfile = \left( u_1^\playerA \! \pureProfile \!, \ldots, u_n^\playerA \! \pureProfile \right). 
\end{align}
Values of vector $\outcomesVector$, i.e. the result of~\eqref{outcomes_vector}, are in $\{\shortminus 1,0,1\}^n$.

\begin{definition}[Aggregation function]\label{agg_func} Function \mbox{$\aggfunc$} is an aggregation function if $f$ returns the same value for all permutations of every vector \mbox{$\outcomesVector \!\in\! \{\shortminus 1,0,1 \} ^n$}.
\end{definition}

\begin{observation}[Aggregation functions]\label{obs:agg_func}
Each function $f$ satisfying Definition \ref{agg_func} can be defined as $f\!\left( \outcomesVector \right) \!=\! f_n\!\left(k_W, k_L\right)$, where $k_W$ and $k_L$ describe the number of 1s and -1s in vector $\outcomesVector$.
\end{observation}
Given the aggregation function $f$ we define the payoffs of each of the players by:
\begin{equation}
\label{pure_payoffs}
    \aPurePayoffArgs = -\bPurePayoffArgs = f\left(u^\playerA \pureProfile\right) .
\end{equation}
As the payoffs of two players always sum up to $0$, the defined class of models describes zero-sum games. Some of the most commonly used aggregation functions are:

\begin{align}
    &f_{blotto}\left(\outcomesVector\right) = \sum_{i=1}^n u_i, \label{blotto_agg} \\
    &f_{mto} \left(\outcomesVector\right) = \sign\left(\sum_{i=1}^n u_i\right), \label{chopstick_agg} \\
    &f_{cho}\left(\outcomesVector\right) = \left[ \sum_{i=1}^n \left[u_i = 1 \right] > n/2 \right] \shortminus \left[ \sum_{i=1}^n \left[u_i = \shortminus 1 \right] > n/2 \right], \label{majoritarian_agg}
\end{align}

where, given a condition $\varphi$, $[\varphi]$ is the Iverson bracket taking value $1$ if $\varphi$ is satisfied and value $0$ otherwise. 

Aggregation function~\eqref{blotto_agg} is the linear function used in Colonel Blotto game~\cite{borel}, \eqref{chopstick_agg} is the ``more than opponent'' function, and \eqref{majoritarian_agg} is the ``majoritarian'' function. 
In general, many more aggregation functions generate nontrivial games.

\begin{definition}[Conflict with multiple battlefields]\label{def:conflict} Quadruple $(\aRes, \bRes, n, f)$ defines a zero-sum two-player game with strategies sets defined by \eqref{pure_strategy} and payoffs defined by \eqref{pure_payoffs}.
\end{definition}
We allow the players to make randomized choices. A \emph{mixed strategy} of player $\examplePlayer$ is a probability distribution on $\thisPureSet$. Given a non-empty set $X$, let $\Delta(X)$ denote the set of all probability distributions on~$X$. The expected payoff to player $\examplePlayer$ from a pair of mixed strategies $(\mixedPair) \in \Delta(\aPureSet) \times \Delta(\bPureSet)$ is equal to
\begin{equation*}
    \thisMixedPayoff = \sum\limits_{(\bm{x}, \bm{z}) \in \aPureSet \times \bPureSet} \xi^\playerA_{\bm{x}} \xi^\playerB_{\bm{z}} \thisPurePayoff \left(\bm{x}, \bm{z} \right).
\end{equation*}
We assume that the players make their decisions ``simultaneously'', i.e. each player chooses her strategy without observing the choice of the opponent. We are interested in (mixed strategy) Nash equilibria (NE) of the game, i.e. mixed strategy profiles $(\mixedPair) \in \Delta(\aPureSet) \times \Delta(\bPureSet)$
such that no player can improve her expected payoff by changing her strategy unilaterally, 
i.e. for each $\examplePlayer$ and all $\bm{\zeta} \in \thisMixedSet$,
\begin{equation*}
\thisMixedPayoff \geq \thisPayoff\left(\bm{\zeta},\bm{\xi}^\thatPlayer\right),    
\end{equation*}
where $\thatPlayer$ denotes the player other than $\thisPlayer$ and $(\bm{\zeta},\bm{\xi}^\thatPlayer)$ is the strategy profile
obtained from $(\mixedPair)$ by replacing $\thisMixed$ with $\bm{\zeta}$. 

\section{Complexity of finding the best pure response}\label{sec:complexity}

Although the complexity of computing Nash equilibria of the games in question (defined by aggregation functions~\eqref{chopstick_agg} and ~\eqref{majoritarian_agg}) remains unknown, we show that finding a pure strategy best response to a given mixed strategy is a computationally hard problem. It is in stark contrast to the Colonel Blotto game, where finding a pure strategy best response is solvable in polynomial time. We prove the result using a reduction from the max coverage problem. This follows the idea in~\cite{BehnezhadBlumDerakhshanHajiaghayiMahdianPapadimitriouRivestSeddighinStark18}, where computing a pure strategy response to a mixed strategy, that maximizes the probability of obtaining payoff above a given threshold in the Colonel Blotto game is shown to be computationally hard.

The hardness result of finding the best response does not necessarily imply any hardness result for the considered game, however, it can be considered an insightful indicator of how hard the game is to solve. See~\cite{Xu2016TheMO} for results about a class of games where the hardness of finding the best response implies the hardness of finding a Nash equilibrium of the game. 

Assume a mixed strategy $\bMixed$ of player $\playerB$, given as a list of pure strategies in its support, and their associated probabilities. The best response problem for a given aggregation function $f$, denoted by $BR_f$, is to find a pure strategy $\aPure$ of player $\playerA$ that maximizes the $\Pi_{f}^{\playerA}(\aPure, \bMixed)$.

\begin{proposition}\label{th:best_response_complexity}
For aggregation functions $f$ given by~\eqref{chopstick_agg} and~\eqref{majoritarian_agg}, there is no polynomial time algorithm to solve $BR_f$ unless $P=NP$.
\end{proposition}

\begin{proof}
    To prove this hardness we provide a reduction from the max-coverage problem to $BR_f$. A number $k$ and a collection of sets $S = \{S_1, S_2, \ldots, S_m\}$ are given. The maximum coverage problem is to find a collection $S' \subseteq S$, such that $|S'| \leq k$ and the number of covered elements (i.e. $|\bigcup_{S_i \in S'}S_i|$) is maximized. 
    
    Let $E = \bigcup_{S_i \in S} S_i$ denote the set of all elements in the given max-coverage instance. Assume, that the number of sets in $S$ in which a given element appears is the same for every element in $E$ and denote this number by $t$. If that is not the case, a polynomial number (with respect to $|E|\cdot|S|)$ of additional singleton sets can be added to collection $S$, without interfering with the max-coverage complexity, or the resulting collection $S'$ as no singleton is more desired than any other sets containing the only element of the considered singleton.

    Consider a game with aggregation function $f$ defined by \eqref{majoritarian_agg}, $2 \cdot |S|$ battlefields, where the player $\playerA$ has $|S| + k$ resources and player $\playerB$ has $(|S| + k + 1) \cdot (|S| - t) \cdot |E|$ resources. The first $|S|$ of $2\cdot |S|$ battlefields correspond to the sets in $S$. For every element $e \in E$, we denote a corresponding pure strategy of player $\playerB$, denoted by $s^{\playerB, e}$, by assigning $|S| + k + 1$ resources to each battlefield that its corresponding set does not contain $e$ and assigning 0 resources in all other battlefields. Assume that player $\playerB$ uses a mixed strategy defined by choosing uniformly from the set of strategies $\bigcup_{e \in E} s^{\playerB, e}$.

    Assume that the best response of player $\playerA$ that maximizes her expected payoff is given. Note, that as player $\playerB$ always assigns either 0 or $|S| + k + 1$ resources to a single battlefield, it is sufficient for the player $\playerA$ to assign either 0 or 1 resources on every battlefield. 
    
    We claim that the best response of player $\playerA$ assigns $|S|$ resources to the remaining $|S|$ battlefields (one to each), which are always assigned 0 by player $\playerB$. This guarantees that player $\playerA$ does not lose with probability 1, as it always wins at least half of the battlefields. The remaining $k$ resources are assigned between the first $|S|$ battlefields, yielding a max-coverage of size $k$ from sets in $|S|$, where each battlefield corresponding to the set in the considered max-coverage will be assigned exactly one resource.

    The proof can be easily adapted for the aggregation function~\eqref{chopstick_agg}. Instead of assigning additional $|S|$ battlefields that are assigned more than 0 by player $\playerB$, we only add $|S| - t$ additional battlefields and we reduce the number of resources of player $\playerA$ from $|S| + k$ to $|S| - t + k$.

\end{proof}

\section{Symmetrized model}
\label{sec:symmetrized}

A player with $\thisRes$ resources to distribute over $n$ battlefields has
\begin{equation}
\label{eq:poly-time}
\binom{n+\thisRes-1}{n-1}  \leq (n+\thisRes-1)^{\min(n-1,\thisRes)}
\end{equation}
pure strategies. Thus, when the number of battlefields or the number of resources of the players are fixed, both players have polynomial size strategy sets and, since the game is zero-sum, a NE can be found in polynomial time~\cite{zero_sum_games_solving}.
It is important to note that the degree of the polynomial that describes the time complexity is limited by the parameter value (the number of battlefields or the number of resources of the two players) which can be arbitrarily large. Moreover, when we consider a model with $n$ battlefields and a fixed parameter $d \in N$, such that the number of resources of both players is $(n+d)$, the number of strategies of the players grows exponentially (c.f.~\eqref{eq:poly-time}). For example, with 20 battlefields and 25 resources, the number of ways in which a player can distribute the resources across the battlefields is more than $10^{12}$.  With such a number of pure strategies, LP solvers, which are standard tools used for solving zero-sum games, become impossible to use in practice.

To address the problem of large strategy sets, we follow the idea of Hart~\cite{hart} and consider a symmetrized variant of the model, defined as follows. Let $P_n$ denote the set of all permutations over $\fieldsSet$. Given a strategy \mbox{$\pure = (s_1, \ldots, s_n) \!\in\! \thisPureSet$} of player $\thisPlayer$ let $\pureSym$ be the mixed strategy which for each permutation $p \in P_n$ over $\fieldsSet$ chooses, with probability $1/n!$, an assignment $p(\pure) = (s_{p(1)}, \ldots,s_{p(n)} )$ of resources. 
\begin{definition} [Symmetrized conflict with multiple battlefields] \label{def:symmetrized_conflicts}
A symmetrized conflict with multiple battlefields is the variant of conflict with multiple battlefields with the same set of players, each player $\examplePlayer$ having a set of strategies $\thisSymSet = \{ \pureSym : \pure \in \thisPureSet\}$ and payoff from strategy profiles 
$(\sigma(\aPure), \sigma(\bPure))\!\!\in\!\!\aSymSet\!\times\!\bSymSet$ 
defined by $\thisPayoff(\sigma(\aPure), \sigma(\bPure))$. 
\end{definition}
The elements of $\thisSymSet$ are called \emph{symmetric strategies}. Since symmetric strategies are order-invariant, when referring to a symmetric strategy $\pureSym$ we will assume that $s_1 \geq \ldots \geq s_n$. Both models (defined in Definitions~\ref{def:conflict} and~\ref{def:symmetrized_conflicts}) describe finite games and therefore, by the Nash theorem, each of these games has a Nash equilibrium in mixed strategies. 

Given a mixed strategy $\mixed \in \thisMixedSet$ of player $\thisPlayer$ in the original game, let $\mixedSym$ be the mixed strategy that for each permutation $p \in P_n$ and each $\pure \in \thisPureSet$ chooses, with probability $\xi_{\pure}/n!$, assignment $(s_{p(1)},\ldots,s_{p(n)})$. 
Notice that for any mixed strategy $\bm{\zeta} \in \Delta(\thisSymSet)$ in the symmetrized game there exists a mixed strategy $\mixed$ in the original game such that 
$\bm{\zeta} = \mixedSym$. We will call the mixed strategies of the symmetrized game \emph{symmetric mixed strategies}.

The following proposition implies that any NE of any game that matches Definition \ref{def:symmetrized_conflicts} is also an NE of a corresponding game described by Definition \ref{def:conflict}.

\begin{proposition}\label{proposition:symm_ne}
For any mixed strategy profile $\mixedProfile$ of the symmetric conflict with multiple battlefields, if $\mixedSymProfile$ is an NE of the symmetrized variant of the game then it is also an NE of the original game.
\end{proposition}

\begin{proof}
First, observe that for any player $\examplePlayer$, any aggregation function $f$, any two mixed strategies $\mixed \in \thisMixedSet$, of $\thisPlayer$ and $\bm{\zeta} \in \thatMixedSet$ of the other player and any permutation $q \in P_n$,
\begin{equation}
\label{eq:symm_payoff}
\begin{aligned}
\thisPayoff(\mixedSym, q(\bm{\zeta}))
& =  \sum_{p \in P_n} \sum_{\bm{x}\in \thisMixedSet}\sum_{\bm{z}\in \thatMixedSet} \frac{\xi_{\bm{x}}}{n!} \zeta_{\bm{z}}  \pi^{\thisPlayer}_f \! \left(p\left(\bm{x}\right),q\left(\bm{z}\right)\right)\\
& =  \sum_{p \in P_n} \sum_{\bm{x}\in \thisMixedSet}\sum_{\bm{z}\in \thatMixedSet} \frac{\xi_{\bm{x}}}{n!} \zeta_{\bm{z}}  \pi^{\thisPlayer}_f \! \left(p\left(q^{-1}\!\left(\bm{x}\right)\right),\bm{z}\right)\\
 =  \sum_{p \in P_n} & \sum_{\bm{x}\in \thisMixedSet} \sum_{\bm{z}\in \thatMixedSet} \frac{\xi_{\bm{x}}}{n!} \zeta_{\bm{z}} \pi^{\thisPlayer}_f \left(p\left(\bm{x}\right),\bm{z}\right) = \thisPayoff\left(\mixedSym, \bm{\zeta}\right).
\end{aligned}
\end{equation}
Take any mixed strategy profile $(\mixedPair)$ of the game in question and suppose that $\mixedSymProfile$ is a MNE of the symmetrized variant of the game. Assume, to the contrary, that it is not a MNE of the original game. It means that there exists a player $\examplePlayer$ and a mixed strategy $\bm{\zeta} \in \thisMixedSet$ of $\thisPlayer$ such that $\thisPayoff\left(\bm{\zeta},\sigma\left(\bm{\xi}^{\thatPlayer}\right)\right) \!\! > \!\! \thisPayoff\mixedSymProfile$. By~\eqref{eq:symm_payoff}, $\thisPayoff\left(\bm{\zeta},\thatMixedSym\right) \!\! = \!\! \thisPayoff\left(\sigma\left(\bm{\zeta}\right),\sigma\left(\thatMixed\right)\right)$. Hence $\thisPayoff\left(\sigma\left(\bm{\zeta}\right),\thatMixedSym\right) \!\! > \!\! \thisPayoff \mixedSymProfile$. Since $\sigma(\bm{\zeta})$ is a symmetric strategy, this contradicts the assumption that the profile $\mixedSymProfile$ is a MNE of the symmetrized variant of the game. Thus the claim of the proposition must hold.

\end{proof}

The advantage of considering the symmetrized models is that the strategy sets of the players are reduced by the exponential factor, which is presented in the following figure.

\begin{figure}
\centering
\begin{subfigure}[t]{0.5\columnwidth}
    \includegraphics[width=1.7in]{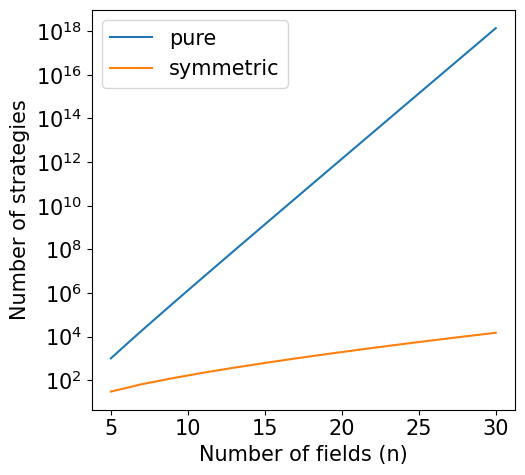}
    \caption{The number of pure and symmetric strategies for $n+5$ resources.}
    \label{subfig:comparison}
    \Description{This shows the comparison of the number of pure and pure symmetrized strategies of a player when the number of resources is equal to the number of battlefields incremented by 5. The range of the number of battlefields is from 5 to 30, and the y-axis scale is logarithmic.}
\end{subfigure}%
\begin{subfigure}[t]{0.5\columnwidth}
    \includegraphics[width=1.7in]{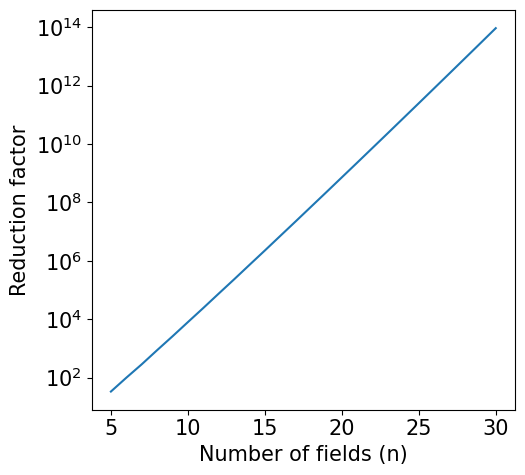}
    \caption{The reduction factor.}
    \label{subfig:reduction}
    \Description{This shows the reduction factor meaning the number of pure strategies divided by the number of pure symmetrized strategies of a player when the number of resources is equal to the number of battlefields incremented by 5. The range of the number of battlefields is from 5 to 30, and the y-axis scale is logarithmic.}
\end{subfigure}
\caption{Comparison of number of pure and pure symmetric strategies}
\label{fig:reduction}
\end{figure}

Figure~\ref{subfig:comparison} shows the comparison of the number of the pure strategies and the number of the symmetric strategies, when the number of resources is a linear function of the number of battlefields (here the number of battlefields increased by 5). The vertical axis shows the number of strategies using a logarithmic scale. By Eq.~\eqref{eq:poly-time}, when only the number of battlefields or only the number of resources grows, the size of strategy sets (both pure and symmetric) grows polynomially. Only when both the parameters grow together at the same time, the size of strategy sets grow exponentially.

Figure~\ref{subfig:reduction} shows the ratio of the number of pure strategies to the number of symmetric strategies (the reduction factor) for data in Figure~\ref{subfig:reduction}. Notice that although the number of symmetric strategies is still exponential, the reduction is by an exponential factor with rescpect to the parameters of the model.

To realize that the number of symmetrized strategies is still exponential with respect to the model parameters when the number of resources of a player is a linear function of the number of the battlefields, note that, as shown in \cite{partitions}, the number $p(n)$ of different partitions of $n$ into the sum of non-negative integers has asymptotic growth of:
\begin{equation}
    p(n) \sim \frac{1}{4n\sqrt{3}} \cdot \exp\left(\pi \sqrt{\frac{n}{2}}\right).
\end{equation}
Consider a model with $n$ battlefields and a fixed parameter $d \in N$, such that the number of resources of both players is $(n+d)$. The number of symmetrized strategies of each player is a number of partitions of $(n+d)$ into at most $n$ parts. This is larger than the number of partitions of $n$ into at most $n$ parts, which is exactly $p(n)$. Therefore, the number of symmetrized strategies grows exponentially with respect to the number of battlefields in such a setting.

\section{Computing payoffs from symmetric strategies}
\label{sec:clash_matrix}
To benefit from restricting to the symmetrized variant of the model, we need to be able to efficiently compute single payoffs from symmetric strategy profiles $\symmProfile$. 

Notice that payoff to player $\examplePlayer$ from a strategy profile $\symmProfile \in \aSymSet \!\times\! \bSymSet$ is given by:
\begin{multline}
    \thisSymPayoff = \frac{1}{\left(n!\right)^2} \sum_{q \in P_n} \sum_{p \in P_n} \thisPurePayoff \! \left(q\left(\aPure\right)\!, p\left(\bPure\right) \right)\! = \\ \frac{1}{n!} \sum_{p \in P_n} \thisPurePayoff \left(\aPure, p\left(\bPure\right) \right).
\end{multline}\label{eq:naive_payoff_2}

To calculate the payoff from a pair of symmetric strategies $\symmProfile$ in a na\"{i}ve way, we need to calculate the payoff of the general model for each permutation. This makes calculating the payoff matrices of the symmetrized games almost as costly as calculating the payoff matrices of unsymmetrized games. To address this issue, we provide an algorithm that computes a single payoff from a pair of symmetric strategies in polynomial time with respect to $n$.

\subsection{Clash matrix}
By a clash matrix of a pair of symmetric strategies $\symmProfile$ we mean a matrix of dimension $n \times n$ defined as follows:
\begin{equation*}
\clashMatrix_{i,j} = \sign\left(s^{\playerA}_i - s^{\playerB}_j\right).
\end{equation*}
That is, the matrix cell at coordinates $(i,j)$ stores information (difference sign) of the result of clashing the $i$-th resource in player $\playerA$'s vector $\aPure$ with the $j$-th resource in player $\playerB$'s vector $\bPure$.

\begin{example}[Clash matrix]
\renewcommand{\kbldelim}{(}
\renewcommand{\kbrdelim}{)}
\[
  \clashM^{(3,1,0),(2,2,0)} = \kbordermatrix{
    & \mathbin{\phantom{-}}2 & \mathbin{\phantom{-}}2 & \mathbin{\phantom{-}}0 \\
    3 & \mathbin{\phantom{-}}1 & \mathbin{\phantom{-}}1 & \mathbin{\phantom{-}}1 \\
    1 & -1 & -1 & \mathbin{\phantom{-}}1 \\
    0 & -1 & -1 & \mathbin{\phantom{-}}0
  }
\]
\end{example}

Note that each permutation in formula \eqref{eq:naive_payoff_2} can be used to select a subset of elements in matrix $\clashMatrix$ in such a way that exactly one element is selected from each row and each column. This is illustrated by the following formula:

\begin{multline}\label{eq:payoff_redef}
    \aSymPayoff = \frac{1}{n!} \sum_{p \in P_n}  \aPurePayoff \left(\aPure, p\left(\bPure\right) \right) = \\
    \frac{1}{n!} \sum_{p \in P_n} f \left(\clashMatrix_{1,p(1)}, \ldots, \clashMatrix_{n,p(n)}  \right) 
\end{multline}

We can think of the problem of computing a single payoff from a pair of symmetric strategies as follows. A matrix $\clashMatrix$ is a chessboard of dimension $n \times n$, divided into three distinct areas:
\begin{align*}
W^{\purePair} & = \{(i,j) \in [n] \times [n] : \clashMatrix_{i,j} = \phantom{-}1\},\\
T^{\purePair} & = \{(i,j) \in [n] \times [n] : \clashMatrix_{i,j} = \phantom{-}0\},\\
L^{\purePair} & = \{(i,j) \in [n] \times [n] : \clashMatrix_{i,j} = -1\}.
\end{align*}
Each permutation $p \in P_n$ corresponds to distributing $n$ rooks on the mentioned chessboard of dimension $n \times n$ in such a way that no two rooks attack each other. Each row contains one rook and the rook in the $i$'th row is located in $p(i)$'th column. For $D \!\in\! \{W,T,L\}$, let $K_D(p \!\mid\! \clashM)$ denote the number of rooks distributed in area $D$ of clash matrix $\clashM$ by permutation $p$.
By Observation~\ref{obs:agg_func}, it follows that:
\begin{align} \label{eq:payoff_with_areas}
    \aPurePayoff \left( \aPure, p\left( \bPure \right)\right)=f_n\!\left( K_W\! \left( p \!\mid\! \clashM \right)\!, K_L\! \left( p \!\mid\! \clashM \right)\right).
\end{align}

Let $h( k_W, k_L \!\mid\! \clashM )$ be the number of permutations $p \in P_n$ that place exactly $k_W$ rooks in area $W$ and $k_L$ rooks in area $L$ for a given clash matrix $\clashM$, i.e.
\begin{align*}
h\!\left( k_W, k_L \!\mid\! \clashM \right) \!=\!\! \sum_{p \in P_n} \!\!\left[ K_W\!\left(p \!\mid\! \clashM \right) \!=\! k_W \land K_L\!\left(p \!\mid\! \clashM\right) \!=\! k_L \right]\!.
\end{align*}
Using function $h$ and \eqref{eq:payoff_with_areas} we can redefine payoff \eqref{eq:payoff_redef}:
\begin{equation}
    \aSymPayoff = \frac{1}{n!} \sum\limits_{k_W = 0}^n \sum\limits_{k_L = 0}^n h\!\left(k_W,k_L\right) f_n\!\left(k_W, k_L\right)\!.
\end{equation}
For better clarity, we omit $\clashMatrix$ parameter in $h$ and $f_n$ notation. In the following subsections, we show how to compute values of function $h$ for a given clash matrix. 

\subsection{Properties of clash matrices}
Note that due to the non-increasing ordering of the strategy vectors of both players, each column and row of a clash matrix describes a monotonic sequence of length $n$. Using this observation, we conclude that all clash matrices share a very particular structure. Area $W$ occupies the upper part of the matrix, potentially reaching lower and lower in each successive column, and area $T$ consists of rectangles, perhaps touching one another at the corners but never sharing a side. Example \ref{ex:cutoffs} illustrates such a matrix.

When looking for the recursive formula for function $h$ we will be ``cutting off'' certain parts of the matrix (submatrix) with a certain number of distributed rooks and thus reducing the size of the considered matrix (submatrix).

\begin{example}[Successive cut-offs of the clash matrix]
\label{ex:cutoffs} 
An example of the clash matrix with the successive ``cut-offs'', marked by their numbers is shown in Figure \ref{fig:clash_matrix}. Areas are colored as follows: yellow ($L$), blue ($T$), and red ($W$). An example of a strategy pair that yields the considered clash matrix is $s_A = (8,8,6,5,4,2,1,1,0), s_B = (8,8,8,7,5,3,3,1,0)$.

\begin{figure}[h]
\begin{center}
\resizebox{150pt}{150pt}{
\begin{tikzpicture}
\filldraw[yellow!50] (0,1) rectangle (9,10);

\filldraw[red!50] (7,2) rectangle (9,10);
\filldraw[red!50] (5,5) rectangle (9,10);
\filldraw[red!50] (3,8) rectangle (9,10);
\filldraw[red!50] (4,7) rectangle (9,10);

\filldraw[cyan!50] (8,1) rectangle (9,2);
\filldraw[cyan!50] (7,2) rectangle (8,4);
\filldraw[cyan!50] (0,8) rectangle (3,10);
\filldraw[cyan!50] (4,6) rectangle (5,7);

\draw[black, very thick] (0,1) rectangle (9,10);



\draw[black, very thick] (0, 2) -- (8,2) -- (8, 10);
\draw[black, very thick] (0, 4) -- (7,4) -- (7, 10);
\draw[black, very thick] (0, 5) -- (7,5);
\draw[black, very thick] (5, 5) -- (5, 10);
\draw[black, very thick] (0, 6) -- (5,6);
\draw[black, very thick] (0, 7) -- (4,7) -- (4, 10);
\draw[black, very thick] (0, 8) -- (4,8);
\draw[black, very thick] (3, 8) -- (3, 10);

\node [anchor=center, scale=2] (note) at (1.5,9) {9};
\node [anchor=center, scale=2] (note) at (4.5,6.5) {6};
\node [anchor=center, scale=2] (note) at (7.5,3) {2};
\node [anchor=center, scale=2] (note) at (8.5,1.5) {1};

\node [anchor=center, scale=2] (note) at (3.5,4.5) {3};
\node [anchor=center, scale=2] (note) at (2.5,5.5) {5};
\node [anchor=center, scale=2] (note) at (2,7.5) {7};
\node [anchor=center, scale=2] (note) at (6,7.5) {4};
\node [anchor=center, scale=2] (note) at (3.5,9) {8};

\node [anchor=center, scale=2] (note) at (0.5, 10.5) {8};
\node [anchor=center, scale=2] (note) at (1.5, 10.5) {8};
\node [anchor=center, scale=2] (note) at (2.5, 10.5) {8};
\node [anchor=center, scale=2] (note) at (3.5, 10.5) {7};
\node [anchor=center, scale=2] (note) at (4.5, 10.5) {5};
\node [anchor=center, scale=2] (note) at (5.5, 10.5) {3};
\node [anchor=center, scale=2] (note) at (6.5, 10.5) {3};
\node [anchor=center, scale=2] (note) at (7.5, 10.5) {1};
\node [anchor=center, scale=2] (note) at (8.5, 10.5) {0};

\node [anchor=center, scale=2] (note) at (-0.5, 9.5) {8};
\node [anchor=center, scale=2] (note) at (-0.5, 8.5) {8};
\node [anchor=center, scale=2] (note) at (-0.5, 7.5) {6};
\node [anchor=center, scale=2] (note) at (-0.5, 6.5) {5};
\node [anchor=center, scale=2] (note) at (-0.5, 5.5) {4};
\node [anchor=center, scale=2] (note) at (-0.5, 4.5) {2};
\node [anchor=center, scale=2] (note) at (-0.5, 3.5) {1};
\node [anchor=center, scale=2] (note) at (-0.5, 2.5) {1};
\node [anchor=center, scale=2] (note) at (-0.5, 1.5) {0};


\end{tikzpicture}
}
\caption{Clash matrix with successive cut-offs}
\label{fig:clash_matrix}
\Description{This figure shows an example of the clash matrix divided into three regions - W, L, and T. Successive cut-offs considered in the clash matrix algorithm are shown and numbered - from 1 to 9. }
\end{center}
\end{figure}
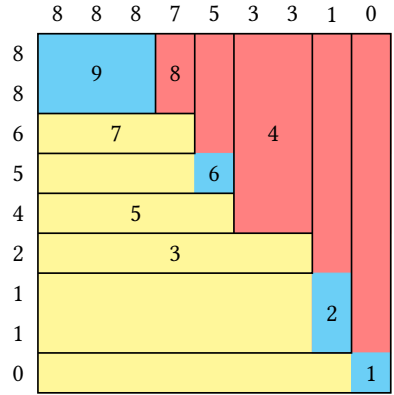
\end{example}

\subsection{Recursive formula}

In this section, we derive a recursive formula describing the number of possible distributions of rooks that do not attack one another on a chessboard (clash matrix) with designated areas $W$, $L$, and $T$. Let:
\begin{itemize}
    \item[] $\clashM_{|(i,j)}$ -- the submatrix of $\clashM$ consisting of the intersection of the first $i$ rows and first $j$ columns of $\clashM$.
    \item[] $H(i,j,m,k_W, k_L \mid \clashM)$ -- the number of ways in which $m$ rooks can be distributed in the intersection of the first $i$ rows and $j$ columns of the clash matrix $\clashM$, such that there are exactly $k_W$ rooks in area $W$ and exactly $k_L$ rooks are in area $L$.
    \item[] $R(i,j,t) = \binom{j}{t} \binom{i}{t} t!$ -- the number of ways to distribute $t$ rooks in a uniform area with $i$ rows and $j$ columns.
\end{itemize}
From the definition of $H$ it follows that 
\begin{equation*}
  h\left(k_W,k_L \mid \clashM \right) = H\left(n,n,n,k_W,k_L \mid \clashM \right).  
\end{equation*}

To find the recursion, in each step, we choose the number of columns and rows so as to cover the entire coherent section of $T$ lying in the right lower corner. If the corner lies in $L$, we say that the width of the considered section of $T$ is $0$. When the corner is in $W$, we say that the height of the considered section of $T$ is $0$. 
The recursive formula is expressed as follows:
\begin{align*}
 H_1 &\left(i,j,m,k_W,k_L \mid \clashM\right) = \\
 & \sum_{\substack{r_1,r_2,r_3 \geq 0 \\ r_1+r_2+r_3 \leq m}}  H(i',j',m-r_{s}, k_W - r_3, k_L - r_1 \mid \clashM_{|(i',j')})  \\
 & \cdot R(i - i', j' - (m - r_{s}), r_1)  \cdot R(i-i'-r_1, j - j', r_2) \cdot \\
 & R(i' - (m-r_{s}), j-j'-r_2, r_3),
\end{align*}
where $r_{s} = r_1 + r_2+r_3$. 
To show the recursion, we consider every possible division of $m$ rooks into four groups, as indicated below in Figure~\ref{fig:cut-offs}. When the height or width of the considered section of $T$ is $0$, two of these groups have a size of $0$.

\begin{figure}[h]
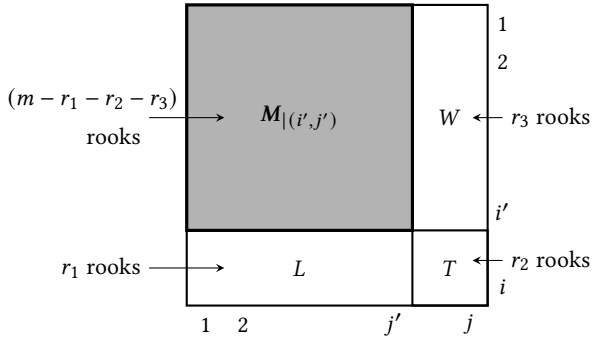

\begin{center}
\subfile{parts/variant_III}  
\caption{Cutting-off procedure}
\label{fig:cut-offs}
\Description{This figure shows a division of m rooks into four regions considered during the cut-off procedure}
\end{center}
\end{figure}

We can think of it as follows:
\begin{enumerate}
    \item We place $(m-r_1-r_2-r_3)$ rooks in area $\clashM_{|(i',j')}$.
    \item In the rectangle lying in $L$ (whose $(m-r_1-r_2-r_3)$ columns have already been excluded in (1)) we place $r_1$ rooks in $i-i'$ free rows and $j' - (m-r_1-r_2-r_3)$ free columns.
    \item In the rectangle lying in $T$ (whose $r_1$ rows have already been excluded in (2)), we place $r_2$ rooks in $i-i'-r_1$ free rows and $j - j'$ free columns.
    \item In the rectangle lying in $W$ (whose $(m-r_1-r_2-r_3)$ rows and $r_2$ columns have already been excluded in (1) and (3)) we place $r_3$ rooks in $i' - (m-r_1-r_2-r_3)$ free rows and $j-j'-r_2$ free columns.
    \item As $r_3$ rooks were placed in $W$ and $r_1$ rooks were placed in $L$, we subtract those values from $k_W$ and $k_L$ respectively in the recursive call of $H$ in point (1). 
\end{enumerate}
Using the formula described above, we will eventually arrive at a situation where submatrix $\clashM$ is entirely within one of the three areas under consideration. Then:
\begin{multline*}
   H_0(i,j,m,k_W,k_L \mid \clashM) = \\  
    \begin{cases}
        R(i,j,m), & \textrm{if $k_W=m$ and $k_L = 0$ and $\clashM \subseteq W$,} \\
        R(i,j,m), & \textrm{if $k_L=m$ and $k_W = 0$ and $\clashM 	\subseteq L$,} \\
        R(i,j,m), & \textrm{if $k_L=k_W = 0$ and $\clashM \subseteq T$,} \\
        0, & \textrm{otherwise.}
    \end{cases}
\end{multline*}

\subsection{Dynamic algorithm}
Values of the recursive formula described above can be computed using dynamic programming using Algorithm~\ref{alg:dynamic}. 
\begin{algorithm}[t]
\SetAlgoLined
\KwIn{Two non-increasing vectors $\aPure$ and $\bPure$}
\KwResult{Payoff value from a given profile.}

\tcc{initialization}
Determine the clash matrix\;
Determine the array $knots[]$, containing pairs of coordinates of consecutive points at which the ``cut-off'' occurs, in ascending order\;
Allocate a $4$-dimensional array $values[]$, of size $2n \times (n+1) \times (n+1) \times (n+1)$\;

 \For{$knotIndex \leftarrow 0$ \KwTo $length(knots) - 1$}
 {
    $knot \leftarrow knots[knotIndex]$\;
    $i \leftarrow knot[0]$\;
    $j \leftarrow knot[1]$\;
    \For{$rooksNum \leftarrow 0$ \KwTo $\min(i,j)$}
    {
        \For{$k_W \leftarrow 0$ \KwTo $\min(i,j)$} 
        {
            \For{$k_L \leftarrow 0$ \KwTo $\min(i,j)$} 
            {
                $values[knotIndex, rooksNum, k_W, k_L] \leftarrow H(i,j,rooksNum,k_W,k_L)$\;
                 \tcc{function H uses $value[]$ and clash matrix}
            }
        }
    }
 }
 $lastKnot \leftarrow length(knots)-1$\; 
 \tcc{last knot is always $(n,n)$}
 $result = 0$\;
 \For{$k_W \leftarrow 0$ \KwTo $n$} 
{
    \For{$k_L \leftarrow 0$ \KwTo $n$} 
    {
        $result += values[lastKnot, n, k_W, k_L] \cdot f_n(k_W, k_L)$
    }
}
 \Return $result/factorial(n)$\;
 \caption{The dynamic algorithm for calculating payoff from a pair of symmetric strategies}\label{alg:dynamic}
\end{algorithm}
Algorithm~\ref{alg:dynamic} runs in memory limited by $O(n^4) = 2n \cdot (n+1) \cdot (n+1) \cdot (n+1)$ and requires no more than $O(n^7) = 4 \cdot 2n \cdot (n+1) \cdot (n+1) \cdot (n+1) \cdot \binom{n+3}{3}$ arithmetic operations.

\section{Double Oracle Algorithm}
\label{sec:doa}
In Section~\ref{sec:clash_matrix} we presented an algorithm for computing single payoffs from symmetric strategy profiles in polynomial time with respect to the model parameters. Sizes of the sets of symmetric strategies of both players, although much smaller, are still exponential with respect to the model parameters. Because of that, LP based methods for solving zero-sum games, which require calculating the whole payoff matrix, are still inefficient for large parameters values. For that reason we use the Double Oracle Algorithm (DOA), described by Algorithm~\ref{alg:doa}, which can be used for finding NE of zero-sum games without calculating the whole payoff matrix. Double Oracle Algorithm correctness was proven in~\cite{double_oracle_algorithm}.

\subsection{Oracle}
When using DOA, one has to use an \textit{oracle} that for a given mixed strategy of a player $\thisPlayer$ returns a best-response pure strategy of the opponent. Given a mixed strategy $\thisMixed$ of player $\thisPlayer$, we consider all the strategy profiles $( \thisMixed, \thatPure )$, where $\thatPure$ is a pure strategy of player $\thatPlayer$, and find a pure strategy that maximizes payoff for player $\thatPlayer$. To compute these payoffs we use Algorithm 1. Given two sets of pure strategies, $X^{\playerA}$ and $X^{\playerB}$, of players $\playerA$ and $\playerB$, respectively, by $\text{CoreLP}(X^{\playerA}, X^{\playerB} )$ we mean a strategy profile that is NE of the zero-sum game restricted to the pure strategies in $X^{\playerA}$ and $X^{\playerB}$, found by an LP-solver.

\begin{algorithm}
\SetAlgoLined
\KwResult{NE of the zero-sum game}

\tcc{initialization}
1. Initialize $X^{\playerA}$ with one pure strategy of player $\playerA$\;
2. Initialize $X^{\playerB}$ with one pure strategy of player $\playerB$\;
 \Repeat{convergence}{
    $\mixedProfile$ $\leftarrow$ CoreLP$\left(X^{\playerA}, X^{\playerB} \right)$\;
    $X^{\playerA} \leftarrow X^{\playerA} \cup \{ best\_response(\bMixed)\}$\;
    $X^{\playerB} \leftarrow X^{\playerB} \cup \{ best\_response(\aMixed) \}$\;
}
 \Return $\mixedProfile$\;
 \caption{The Double Oracle Algorithm for solving zero-sum games}\label{alg:doa}
\end{algorithm}
As we can choose the initial strategies of both players arbitrarily, we decided to start with the most even assignment between all the battlefields for each player.

\subsection{Proposed heuristic}
\label{sec:heuristic}

In this section we propose a simple, but effective, heuristic that allows us to avoid unnecessary calculations for a subset of strategies.
We first prove a simple observation regarding the pure strategies that are the best responses to a mixed strategy of the opponent. By a \emph{best response} of player $\thisPlayer$ to a mixed strategy $\thatMixed$ of the opponent we mean a pure strategy $\thisPure_{best}$ that maximizes the payoff of $\thisPlayer$ against $\thatMixed$.
\begin{definition}[Maximal assignment]
We define $\maxass(\thisMixed) \in N$, of player $\thisPlayer$ as
\begin{equation}
\max_{\substack{s \in \thisPureSet \\ \thisMixed\left(s\right) > 0}} \max_{i \in \{1,\ldots,n\}} s_i,
\end{equation}
i.e. the biggest number of resources that player $\thisPlayer$ assigns to single battlefield with positive probability when playing the strategy $\thisMixed$.
\end{definition}

Let $f$ be an aggregation function that is monotonic non-decreasing in its first argument and monotonic non-increasing in its second argument, meaning that a player cannot decrease her payoff by increasing the number of battlefields won and cannot increase her payoff when the number of battlefields won by the opponent increases. When the disproportion of resources between the players, meaning $|\aRes - \bRes|$ is not greater then $n$, the following proposition always holds. 
\begin{proposition} \label{proposition:max_assign_response}
For any mixed strategy $\thisMixed$ of player $\thisPlayer$ and any aggregation function $f$ as described above, there exists a best response pure strategy $\thatPure_{best}$ of player $\thatPlayer$, that satisfies
\begin{equation}\label{maxass_theorem}
    \maxass(\thatPure_{best}) \leq \maxass(\thisMixed) + 1.
\end{equation}
\end{proposition}

\begin{proof}
Take any pure strategy $\thatPure$ that maximizes payoff of player $\thatPlayer$ against a given mixed strategy $\thisMixed$. If this strategy does not satisfy~\eqref{maxass_theorem} then we can reduce the number of resources at battlefields where the number of resources exceeds $(\maxass(\thisMixed) + 1)$ to $(\maxass(\thisMixed) + 1)$, as those battlefields are still always won by $\thatPlayer$, and increase the number of resources at the battlefields where the current assignment is below $(\maxass(\thisMixed) + 1)$. As the aggregation function is monotonic non-decreasing in its first argument and monotonic non-increasing in its second argument, the payoff of player $\thatPlayer$ does not decrease after such a redistribution. 
\end{proof}

By Proposition \ref{proposition:max_assign_response}, when computing a best response to a given mixed strategy $\thisMixed$, we can restrict attention to the pure strategies of the opponent for which the $\maxass$ value exceeds the $\maxass$ value of the mixed strategy $\xi^C$ by at most one. This significantly reduces the number of necessary calculations and therefore the time required to compute the NE of the game in question.

\section{Experimental evaluation} 
\label{sec:evaluation}
We implemented two versions of Algorithm~1 for finding payoff matrices of the symmetrized games with the aggregation functions, one given by~\eqref{majoritarian_agg} and one given by~\eqref{chopstick_agg}, using $C\texttt{++}$'s \emph{vector} class from $C\texttt{++}$'s standard library. We compared the acquired run times of the algorithm with CPU and GPU-based approaches that calculate payoffs based on \eqref{eq:naive_payoff_2} for finding payoff matrices of the symmetrized game. Experiments were run on a computational cluster XXXX (anonymization) with a central processing unit (CPU) Intel Xeon E5-2640 v4 (10 cores) and the GPU unit Titan V. Four cores of the processor were used while conducting each of the experiments. All the time results in the experiments are averages from $10$ runs of the program. The program we used to obtain payoff matrices with CPU and GPU is proposed and described in \cite{litwin}.
\begin{figure}
\centering
\begin{subfigure}{0.5\columnwidth}
    \includegraphics[width=1.69in]{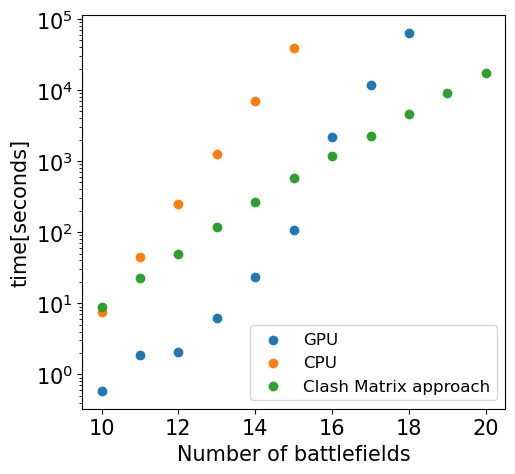}
    \caption{More than opponent}
    \label{subfig:runtimes_mto}
    \Description{This figure provides the comparison of runtimes of different algorithms used for calculating the entire payoff matrix of a symmetrized "more than opponent" model. The yellow dots provide the runtime of a na\"{i}ve CPU algorithm, the blue dots provide the runtime of a na\"{i}ve algorithm using GPU and the green dots provide the runtime of a Clash matrix algorithm. The y-axis scale is logarithmic}
\end{subfigure}%
\begin{subfigure}{0.5\columnwidth}
    \includegraphics[width=1.69in]{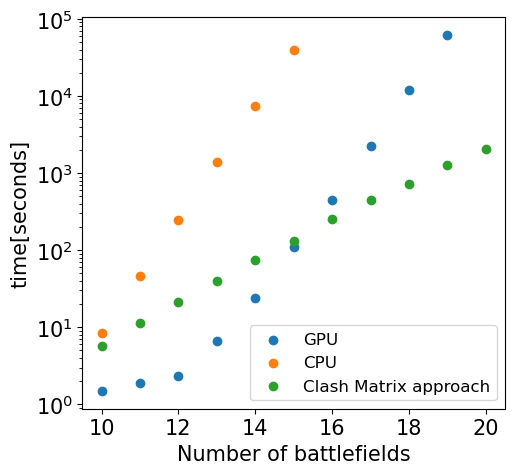}
    \caption{Majoritarian}
    \label{subfig:runtimes_maj}
    \Description{This figure compares runtimes of different algorithms used for calculating the entire payoff matrix of a symmetrized "majoritarian" model. The yellow dots provide the runtime of a na\"{i}ve CPU algorithm, the blue dots provide the runtime of a na\"{i}ve algorithm using GPU and the green dots provide the runtime of a Clash matrix algorithm. The y-axis scale is logarithmic}
\end{subfigure}
\caption{Run times to the number of battlefields}
\label{fig:runtimes}
\end{figure}
Figure \ref{fig:runtimes} shows the comparison of the times required by all three methods for finding the whole payoff matrices of both games. The numbers of battlefields are shown on the horizontal axis. The number of resources of each player is a linear function of the number of battlefields (the number of battlefields increased by $5$). The vertical axis shows the run times of the programs using a logarithmic scale. The time limit for each experiment was set to $10^5$ seconds. 

When considering the \textit{more than opponent} (Figure \ref{subfig:runtimes_mto}) game, the maximal value of parameters that GPU based approach was able to calculate the payoff matrix for was $23$ resources to distribute over 18 battlefields. The achieved speedup of the clash matrix for the same parameters was 18 times. For the \textit{majoritarian} (Figure \ref{subfig:runtimes_maj}) game, the maximal value of parameters that GPU based approach was able to calculate the payoff matrix for was $24$ resources to distribute over 19 battlefields. The achieved speedup of the clash matrix for the same parameters was 48 times. The speedup for $23$ resources to distribute over 18 battlefields was 16 times (similar for both aggregation functions). 

Although our method is significantly faster, the time required to calculate the payoff matrix is still the bottleneck for solving the models in question. Therefore we include Figure~\ref{fig:runtimes_algs}, which shows a comparison of the time results of the calculation of the entire payoff matrix of the game, with the two versions of DOA (with and without the proposed heuristic) that yield a Nash Equilibrium of a game. We use the time results for calculating the entire payoff matrix as a time-bound on every approach that calculates the entire payoff matrix and then solves the game using this matrix (e.g. the standard LP-solver approach). All of the programs use a symmetrized model and calculate single payoffs using the clash matrix. The numbers of battlefields are shown on the horizontal axis. The number of resources of each player is once again a linear function of the number of battlefields (the number of battlefields increased~by~5). We used the state-of-the-art \emph{Gurobi} \cite{gurobi} LP-solver for solving the matrix games as well as finding the NE when using DOA ($\text{CoreLP}$ function).

For the \textit{more than opponent} game (Figure \ref{subfig:alg_mto}), speed-up achieved for the game where both players have 25 resources to distribute over 20 battlefields by the DOA algorithm when compared to the LP-solver is 34 times. In contrast, when using the DOA algorithm with the heuristic that we propose, one gets a speed-up of more than 300 times. When considering the \textit{more than opponent} game (Figure \ref{subfig:alg_maj}), speed-up achieved for the game where both players have 25 resources to distribute over 20 battlefields by the DOA algorithm when compared to the LP-solver is 10 times. In contrast, when using the DOA algorithm with the heuristic that we propose, one gets a speed-up of more than 80 times. This shows that the proposed heuristic, although simple, is very effective in both models. When comparing the LP-solver-based approach, where payoff matrices are computed using a GPU, and the clash matrix-based DOA approach, using the proposed heuristic, the achieved speed-up for the game where both players have 25 resources to distribute over 20 battlefields would be more than 3000 times for both models. Unfortunately, the run time of the GPU-based method for finding the payoff matrix takes too long for us to know the exact speedups that could be achieved.

\begin{figure}[!t]
\centering
\begin{subfigure}{0.5\columnwidth}
    \includegraphics[width=1.69in]{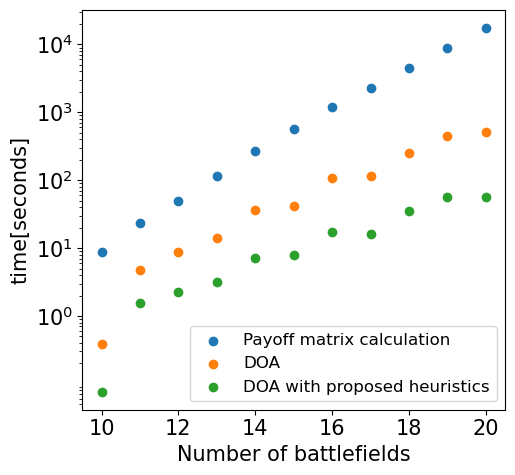}
    \caption{More than opponent}
    \label{subfig:alg_mto}
    \Description{This figure compares runtimes of different algorithms used for calculating the Nash Equilibrium of a symmetrized "more than opponent" model. The blue dots provide the runtime of an LP-solver algorithm, the orange dots provide the runtime of the double oracle algorithm and the green dots provide the runtime of the double oracle algorithm using the proposed heuristic, that allows for narrowing the set of pure strategies in which the best response if found. The y-axis scale is logarithmic}
\end{subfigure}%
\begin{subfigure}{0.5\columnwidth}
    \includegraphics[width=1.69in]{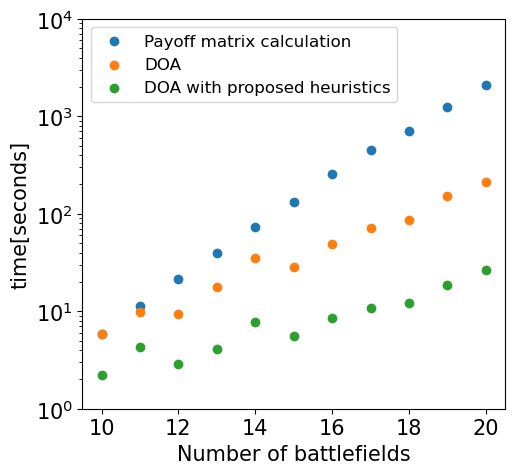}
    \caption{Majoritarian}
    \label{subfig:alg_maj}
    \Description{This figure compares runtimes of different algorithms used for calculating the Nash Equilibrium of a symmetrized "majoritarian" model. The blue dots provide the runtime of an LP-solver algorithm, the orange dots provide the runtime of the double oracle algorithm and the green dots provide the runtime of the double oracle algorithm using the proposed heuristic, that allows for narrowing the set of pure strategies in which the best response if found. The y-axis scale is logarithmic}
\end{subfigure}
\caption{Run times of different algorithms}
\label{fig:runtimes_algs}
\end{figure}

\section{Conclusions}
\label{sec:concl}
In this paper, we consider a class of conflicts with multiple battlefields with discrete resources and uniform battlefields values. We show that the NE of the models in question can be found by examining the symmetrized models, where the number of strategies, although still exponential with respect to the model parameters, is reduced by an exponential factor.

Our contribution has two main components. First, we propose a \emph{clash matrix algorithm} to obtain a single payoff for a pair of symmetric strategies in polynomial time. We carry out a comparison of run times required to compute the payoff matrices using the clash matrix algorithm and na\"{i}ve algorithms (CPU and GPU based). Our algorithm, which works for the whole class of the described models, is significantly faster in practice than both mentioned approaches. The symmetrization combined with the proposed clash matrix algorithm allows for reducing the number of strategies by an exponential factor at the expense of polynomial time for computing the payoffs.

Second, we apply the Double Oracle Algorithm to find an NE of the Chopstick Auctions. Using DOA with the heuristic that we propose, combined with the clash matrix approach for calculating single payoffs gives significant speed-up for the considered subclass of games. 

\begin{acks}
This work was supported by the Polish National Science Centre through grant 2018/29/B/ST6/00174. 
\end{acks}







\bibliographystyle{ACM-Reference-Format} 
\bibliography{paper}


\end{document}


\section{Recurrence formula extended explanation}
\subsection{Features of the clash matrix}
Note that due to the non-increasing ordering of the strategy vectors of both player, each column and row of a clash matrix describes monotonic sequence of length $n$. Two possible variants of the relation resulting from the above description for a pair of adjacent columns can be seen below:


Using this observation, we conclude that the clash matrix has a very characteristic structure. Area $W$ occupies the upper part of the matrix, potentially reaching lower and lower in each successive column, and area $T$ consists of rectangles, perhaps touching one another at the corners but never sharing a side. Example \ref{ex:cutoffs} illustrates such a matrix.

Note that clash matrices (or their submatrices) can be divided into three categories depending on what area the bottom right corner of the matrix (submatrix) is in. When looking for the recursive formula we will be ``cutting off'' certain parts of the matrix (submatrix) with a certain number of distributed rooks and thus reducing the size of the considered matrix (submatrix).

Using the aforementioned features of the clash matrix, we can easily indicate the successive ``cut-offs'' that reduce the matrix (submatrix) and allow us to describe the recursive relation. 

\begin{example}[Successive cutoffs of the clash matrix]\label{ex:cutoffs} 
An example of the clash matrix with the successive ``cut-offs'', marked by their numbers is shown in Figure \ref{fig:clash_matrix}. Area are colored as follows -- $L$ is yellow, $T$ is blue and $W$ red.

\begin{figure}[h]
    \caption{Clash matrix with successive cut-offs}
\begin{center}
\resizebox{180pt}{180pt}{
\begin{tikzpicture}
\filldraw[yellow!50] (0,1) rectangle (9,10);

\filldraw[red!50] (7,2) rectangle (9,10);
\filldraw[red!50] (5,5) rectangle (9,10);
\filldraw[red!50] (3,8) rectangle (9,10);
\filldraw[red!50] (4,7) rectangle (9,10);

\filldraw[cyan!50] (8,1) rectangle (9,2);
\filldraw[cyan!50] (7,2) rectangle (8,4);
\filldraw[cyan!50] (0,8) rectangle (3,10);
\filldraw[cyan!50] (4,6) rectangle (5,7);

\draw[black, very thick] (0,1) rectangle (9,10);



\draw[black, very thick] (0, 2) -- (8,2) -- (8, 10);
\draw[black, very thick] (0, 4) -- (7,4) -- (7, 10);
\draw[black, very thick] (0, 5) -- (7,5);
\draw[black, very thick] (5, 5) -- (5, 10);
\draw[black, very thick] (0, 6) -- (5,6);
\draw[black, very thick] (0, 7) -- (4,7) -- (4, 10);
\draw[black, very thick] (0, 8) -- (4,8);
\draw[black, very thick] (3, 8) -- (3, 10);

\node [anchor=center, scale=2] (note) at (1.5,9) {9};
\node [anchor=center, scale=2] (note) at (4.5,6.5) {6};
\node [anchor=center, scale=2] (note) at (7.5,3) {2};
\node [anchor=center, scale=2] (note) at (8.5,1.5) {1};

\node [anchor=center, scale=2] (note) at (3.5,4.5) {3};
\node [anchor=center, scale=2] (note) at (2.5,5.5) {5};
\node [anchor=center, scale=2] (note) at (2,7.5) {7};
\node [anchor=center, scale=2] (note) at (6,7.5) {4};
\node [anchor=center, scale=2] (note) at (3.5,9) {8};

\node [anchor=center, scale=2] (note) at (0.5, 10.5) {8};
\node [anchor=center, scale=2] (note) at (1.5, 10.5) {8};
\node [anchor=center, scale=2] (note) at (2.5, 10.5) {8};
\node [anchor=center, scale=2] (note) at (3.5, 10.5) {7};
\node [anchor=center, scale=2] (note) at (4.5, 10.5) {5};
\node [anchor=center, scale=2] (note) at (5.5, 10.5) {3};
\node [anchor=center, scale=2] (note) at (6.5, 10.5) {3};
\node [anchor=center, scale=2] (note) at (7.5, 10.5) {1};
\node [anchor=center, scale=2] (note) at (8.5, 10.5) {0};

\node [anchor=center, scale=2] (note) at (-0.5, 9.5) {8};
\node [anchor=center, scale=2] (note) at (-0.5, 8.5) {8};
\node [anchor=center, scale=2] (note) at (-0.5, 7.5) {6};
\node [anchor=center, scale=2] (note) at (-0.5, 6.5) {5};
\node [anchor=center, scale=2] (note) at (-0.5, 5.5) {4};
\node [anchor=center, scale=2] (note) at (-0.5, 4.5) {2};
\node [anchor=center, scale=2] (note) at (-0.5, 3.5) {1};
\node [anchor=center, scale=2] (note) at (-0.5, 2.5) {1};
\node [anchor=center, scale=2] (note) at (-0.5, 1.5) {0};


\end{tikzpicture}
}\label{fig:clash_matrix}
\end{center}
\end{figure}
\end{example}
\subsection{Recursive formula}

In this section we derive a recursive formula describing the number of possible distributions of non-checking rooks on a chessboard (clash matrix) with designated areas $W$, $L$, and $T$. Let:
\begin{itemize} 
    \item[] $\clashM_{|(i,j)}$ -- the submatrix of $\clashM$ consisting of the intersection of the first $i$ rows and $j$ columns of $\clashM$.
    \item[] $H(i,j,m,k_W, k_L \mid \clashM)$ -- the number of different ways in which $m$ rooks can be distributed in the intersection of the first $i$ rows and $j$ columns of the clash matrix $\clashM$, such that there are exactly $k_W$ rooks in area $W$ and exactly $k_L$ rooks are in area $L$.
    \item[] $R(i,j,t) = {j \choose t} {i \choose t} t!$ -- the number of different ways to place $t$ rooks in a uniform area with $i$ rows and $j$ columns.
\end{itemize}
From the definition of $H$ it follows that $h\left(k_W,k_L \mid \clashM \right) = H\left(n,n,n,k_W,k_L \mid \clashM \right)$.

As mentioned before, we can categorize clash matrices (submatrices) into three categories and therefore we have to define three variants of recursion. In the first variant, we want to ``cut off'' all rows that are entirely in $L$, in the second variant all columns that are in $W$. In the third variant, we choose the number of columns and rows so as to cover the entire coherent section of $T$.

\subsection{Corner in $L$ or $W$ (Variant 1 and 2)}
We show how to describe the recursion in variant one, when the bottom right corner of the fragment under consideration is in area $L$. Note that in such a case, the following is true:
\begin{multline*}
    H_1(i,j,m,k_W,k_L \mid \clashM) =\\ \sum_{r = 0}^m H(i',j,m,k_W, k_L-r \mid \clashM_{|(i',j)}) \cdot R(i - i', j - (m-r), r),
\end{multline*}
where $i'$ is the index of the last row that does not belong entirely to area $L$. This can be understood as follows. We divide $m$ rooks into two groups with corresponding counts $(m-r)$ and $r$, as shown below:
\begin{center}
    \documentclass[../paper]{subfiles}
\begin{document}
\begin{tikzpicture}
\draw[black, thick] (0,0) rectangle (4,4);
\filldraw[color=black, fill=black!30, very thick] (0,1) rectangle (4,4);
\node [anchor=center] (note) at (2,.5) {$W$};
\node [anchor=center] (note) at (2,2.5) {$\bm{M}_{|(i',j)}$};

\node [anchor=center] (note) at (4.2,3.75) {$1$};
\node [anchor=center] (note) at (4.2,3.25) {$2$};
\node [anchor=center] (note) at (4.2,1.25) {$i'$};
\node [anchor=center] (note) at (4.2,0.25) {$i$};

\node [anchor=center] (note) at (0.25,-0.25) {$1$};
\node [anchor=center] (note) at (0.75,-0.25) {$2$};
\node [anchor=center] (note) at (3.75,-0.25) {$j$};

\draw [-stealth](-0.5,0.5) -- (0.5,0.5);
\draw [-stealth](-0.5,2.5) -- (0.5,2.5);

\node [anchor=east] (note) at (-0.5,0.5) {$r$ rooks};
\node [anchor=east] (note) at (-0.5,2.5) {$(m-r)$ rooks};
\end{tikzpicture}
\end{document}
\end{center}

 The first group is distributed in the submatrix $\clashM_{|(i',j)}$, and the second group is distributed in the lower rectangle that lies entirely in $L$. The $(m-r)$ of columns in this rectangle are already inaccessible (checked) by the rooks deployed in the upper part, hence $R(i - i', j - (m-r), r)$ describes the number of possible deployments in this rectangle. All the rooks deployed in lower rectangle are in $L$, therefore the $k_W$ parameter has to stay the same for upper submatrix, while $k_L$ is decreased by $r$. 
 
 Second variant is analogous to the first one. Only difference is that in this variant we decrease $k_W$ by $r$ as this number of rooks is placed in rectangle that is in $W$. Therefore, we get:
\begin{multline*}
H_2(i,j,m,k_W,k_L \mid \clashM) = \\
\sum_{r = 0}^m H(i,j',m-r,k_W-r,k_L \mid \clashM_{|(i,j')}) \cdot R(i - (m-r), j - j', r),
\end{multline*}
where j' is the index of the last column that does not belong entirely to area $W$.

\subsection{Corner in $T$ (Variant 3)}
In variant 3, the recursive formula is expressed as follows:
\begin{align*}
H_3(i,j,m,&k_W,k_L \mid \clashM) = \\ 
\sum_{\substack{r_1,r_2,r_3 \geq 0 \\ r_1+r_2+r_3 \leq m}} & H(i',j',m-r_{s}, k_W - r_3, k_L - r_1 \mid \clashM_{|(i',j')}) \cdot \\ 
\cdot & R(i - i', j' - (m - r_{s}), r_1)  \cdot R(i-i'-r_1, j - j', r_2) \cdot\\
\cdot & R(i' - (m-r_{s}), j-j'-r_2, r_3),
\end{align*}
where $r_{s} = r_1 + r_2+r_3$. 
Analogous to the simpler variants described above, this time  we divide $m$ rooks into four groups, as indicated below:
\begin{center}
\documentclass[../paper]{subfiles}
\begin{document}
\begin{tikzpicture}
\draw[black, thick] (0,0) rectangle (4,4);
\draw[black, thick] (3,0) rectangle (4,1);
\filldraw[color=black, fill=black!30, very thick] (0,1) rectangle (3,4);
\node [anchor=center] (note) at (3.5,.5) {$T$};
\node [anchor=center] (note) at (1.5,.5) {$L$};
\node [anchor=center] (note) at (3.5,2.5) {$W$};
\node [anchor=center] (note) at (1.5,2.5) {$\bm{M}_{|(i',j')}$};

\node [anchor=center] (note) at (4.2,3.75) {$1$};
\node [anchor=center] (note) at (4.2,3.25) {$2$};
\node [anchor=center] (note) at (4.2,1.25) {$i'$};
\node [anchor=center] (note) at (4.2,0.25) {$i$};

\node [anchor=center] (note) at (0.25,-0.25) {$1$};
\node [anchor=center] (note) at (0.75,-0.25) {$2$};
\node [anchor=center] (note) at (2.75,-0.25) {$j'$};
\node [anchor=center] (note) at (3.75,-0.25) {$j$};

\draw [-stealth](-0.5,0.5) -- (0.5,0.5);
\draw [-stealth](-0.5,2.5) -- (0.5,2.5);

\draw [-stealth](4.2,0.6) -- (3.8,0.6);
\draw [-stealth](4.2,2.5) -- (3.8,2.5);

\node [anchor=east] (note) at (-0.5,0.5) {$r_1$ rooks};
\node [anchor=west] (note) at (4.2,0.6) {$r_2$ rooks};
\node [anchor=west] (note) at (4.2,2.5) {$r_3$ rooks};
\node [anchor=east] (note) at (0,2.75) {$(m-r_1-r_2-r_3)$};
\node [anchor=east] (note) at (-0.5,2.25) {rooks};
\end{tikzpicture}   
\end{document}    
\end{center}

We can think of it as follows:
\begin{enumerate}
    \item We place $(m-r_1-r_2-r_3)$ rooks in area $\clashM_{|(i',j')}$.
    \item In the rectangle lying in $L$ (whose $(m-r_1-r_2-r_3)$ columns have already been excluded in (1)) we place $r_1$ rooks in $i-i'$ free rows and $j' - (m-r_1-r_2-r_3)$ free columns.
    \item In the rectangle lying in $T$ (whose $r_1$ rows have already been excluded in (2)), we place $r_2$ rooks in $i-i'-r_1$ free rows and $j - j'$ free columns.
    \item In the rectangle lying in $W$ (whose $(m-r_1-r_2-r_3)$ rows and $r_2$ columns have already been excluded in (1) and (3)) we place $r_3$ rooks in $i' - (m-r_1-r_2-r_3)$ free rows and $j-j'-r_2$ free columns.
\end{enumerate}
\subsection{Boundary conditions} 
Using the formula described above, we will at some point arrive at a situation where submatrix $\clashM$ is entirely within one of the three areas under consideration. Then:
\begin{align*}
   &  H_4(i,j,m,k_W,k_L \mid \clashM) = \\ & = 
    \begin{cases}
        R(i,j,m), & \textrm{if $k_W=m$ and $k_L = 0$ and $\clashM \in W$,} \\
        R(i,j,m), & \textrm{if $k_L=m$ and $k_W = 0$ and $\clashM \in L$,} \\
        R(i,j,m), & \textrm{if $k_L=k_W = 0$ and $\clashM \in T$,} \\
        0, & \textrm{otherwise.}
    \end{cases}
\end{align*}

\section{Pessimistic costs of Algorithm 1}
\documentclass[../paper]{subfiles}
\begin{document}
Consider Algorithm~\ref{alg:dynamic}. Note that the array \texttt{knots[]} has no more than $2n$ elements (coordinate pairs). This is because the first element certainly has at least one coordinate greater than $0$ (otherwise it would describe a trivial rectangle) and with each subsequent node at least one coordinate increases by at least $1$.

For each of no more than $2n$ knots, Algorithm~\ref{alg:dynamic} checks no more than $(n+1)$ different numbers of rooks, no more than $(n+1)$ different values of the parameter $k_W$, and no more than $(n+1)$ different values of the parameter $k_L$, hence the required memory of Algorithm~\ref{alg:dynamic} is limited by $O(n^4) = 2n \cdot (n+1) \cdot (n+1) \cdot (n+1)$. 
 
For each element of the array \textit{values}[], Algorithm~\ref{alg:dynamic} computes the corresponding sum. The most time-expensive case is Variant 3, where the corner of the considered matrix (submatrix) lies in $T$ which requires summing over three iterators $(r_1, r_2, r_3)$. In this case there are as many different elements of the sum as there are divisions of the number of rooks under consideration into the sum of $4$ natural numbers, i.e., $m+3 \choose 3$. We also perform $3$ multiplications and one addition for each division. This means that the Algorithm~\ref{alg:dynamic} does not perform more than $4 \cdot 2n \cdot (n+1) \cdot (n+1) \cdot (n+1) \cdot {n+3 \choose3} = O(n^7)$ arithmetic operations.
\end{document}
\section{MWU algorithm}
The algorithm, given in Algorithm 2 and described in~\cite{mwu_meta}, requires the following arguments:

\begin{description}
    \item[\textit{matrix}] -- the payoff matrix of the column player of the given game. We obtain the payoff matrix using Algorithm~1.
    \item[\textit{phi}] -- a multiplier in the range $(0,0.5]$, which is the parameter by which we multiply the cost of each decision before updating its weight. 
    \item[\textit{stepsNumber}] -- the number of steps, or iterations of the main loop of the algorithm. As the number of steps increases, the accuracy of the results should increase.
\end{description}
Function $normalize$, used at the end of the main loop, is defined as follows:
\begin{equation*}
    normalize(\bm{x}) = \frac{1}{\sum_{i=1}^n x_i} \cdot \bm{x}.
\end{equation*}

\begin{algorithm}
\SetAlgoLined 
\KwIn{Parameters $matrix$, $phi$, and $stepsNumber$}
\KwResult{A pair of mixed strategies that are $\epsilon$-equilibrium of a given game}

 Prepare a vector $p_t$ of size equal to the number of pure symmetric row player strategies with uniform distribution\;
 Prepare a null vector $jSummed$ of size corresponding to the number of pure symmetric strategies of the column player\;
 $smallestColumnPayoff = 1$\;
 $pBest \leftarrow p_t$\;
 \For{$step \leftarrow 0$ \KwTo $stepsNumber - 1$}
 {
    $columnPlayerPayoffs \leftarrow p_t \cdot matrix$\;
    $bestResponse \leftarrow \arg\max(columnPlayerPayoffs)$\;
    \If{$columnPlayerPayoffs[bestResponse] < smallestColumnPayoff$}
    {
        $smallestColumnPayoff \leftarrow columnPlayerPayoffs[bestResponse]$\;
        $pBest \leftarrow p_t$\;
    }
    $jSumed[bestResponse] \leftarrow  jSumed[bestResponse] + 1$\;
    $m_t \leftarrow matrix[:,bestResponse]$\;
    \tcc{assign bestResponse'th column of matrix as costs vector}
    $p_t \leftarrow (1 - phi \cdot m_t) \cdot p_t$\;
    \tcc{elementwise multiplication of two vectors}
    $normalize(p_t)$\;
 }
 $jDistribution \leftarrow normalize(jSumed)$\;
 \Return{$(pBest, jDistribution)$}
 \caption{Multiplicative weights update algorithm}
\end{algorithm}